\documentclass[a4paper,10pt]{article}%

\usepackage[longnamesfirst]{natbib}
\usepackage{enumerate}
\usepackage{amsfonts}
\usepackage{amsthm}
\usepackage{amsmath}
\usepackage{needspace}
\usepackage{xcolor}

\usepackage[colorlinks=true,citecolor=black, linkcolor=black, bookmarks={false}, pdfpagelayout={SinglePage}  ]{hyperref}

\usepackage{tabularx}
\usepackage{tikz}
\usepackage{amssymb}
\usepackage{graphicx}%
\DeclareMathOperator{\conv}{conv}
\DeclareMathOperator{\dom}{dom}
\DeclareMathOperator{\successors}{succ}
\DeclareMathOperator{\cone}{cone}
\DeclareMathOperator{\interior}{int}

\DeclareMathOperator{\epi}{epi}

\newtheorem{theorem}{Theorem}

\newtheorem{proposition}[theorem]{Proposition}

\newtheorem{remark}[theorem]{Remark}
\newtheorem{example}[theorem]{Example}
\numberwithin{equation}{section}
\numberwithin{theorem}{section}

\newlength{\figurewidth}
\newlength{\figureheight}
\tikzstyle{axis}=[thin]
\tikzstyle{function}=[thick]
\usetikzlibrary{plotmarks}

\usepackage{dcolumn}
\newcolumntype{d}{@{}D{.}{.}{3}@{}}

\let\marginparnew=\marginpar
\long\def\marginpar#1{\marginparnew{\tiny #1}}

\begin{document}

\title{Linear vector optimization and European option pricing under proportional transaction costs}
\author{Alet Roux\thanks{Department of Mathematics, University of York, Heslington,
YO105DD, United Kingdom. Email: alet.roux@york.ac.uk.}
\and Tomasz Zastawniak\thanks{Department of Mathematics, University of York,
Heslington, YO105DD, United Kingdom. Email: tomasz.zastawniak@york.ac.uk.}}
\maketitle

\begin{abstract}%
A method for pricing and superhedging European options under proportional
transaction costs based on linear vector optimisation and geometric duality
developed by \cite{Loehne_Rudloff2014} is compared to a special case of the
algorithms for American type derivatives due to \cite{Roux_Zastawniak2014}. An
equivalence between these two approaches is established by means of a general
result linking the support function of the upper image of a linear vector
optimisation problem with the lower image of the dual linear optimisation
problem.%
\end{abstract}

\section{Introduction}

We compare two existing methods for the computational pricing and superhedging
of European options in the presence of proportional transaction costs, and
investigate the relationships between them, highlighting their similarities,
differences and relative strengths. One of these methods, based on the primal
and dual constructions stated in Section~\ref{Sect:primal-dual-constr}, goes
back to \cite{Roux_Tokarz_Zastawniak2008} and \cite{Roux_Zastawniak2014},
where it was developed for the much more general class of American type
derivative securities, of which European options are a special case. The other
method, which relies on linear vector optimisation and geometric duality, was
proposed by \cite{Loehne_Rudloff2014} and named the SHP-algorithm by them; see
Section~\ref{Sect:SHP-alg}.

As a by-product, we prove a general result establishing one-to-one
correspondence between the support function of the upper image of a linear
vector optimisation problem on the one hand, and the lower image of the dual
linear vector optimisation problem on the other hand; see
Proposition~\ref{Prop1-ver2}. This result provides a link between the two
methods for pricing and superhedging European options, and it is also
interesting in its own right.

We work within the general model of a currency exchange market of \cite{kabanov1999}, with proportional transaction costs included in the form
of exchange rate bid ask spreads. This model has been extensively studied, for
example, by \cite{kabanov_stricker2001b}, \cite{kabanov_rasonyi_stricker2002}
and \cite{schachermayer2004}.

All three algorithms, the primal construction, the dual construction and the
SHP-algorithm lend themselves well to computer implementation. For the primal
and dual constructions this has been done by \cite{Roux_Zastawniak2014} with
the aid of the \emph{Maple} package \emph{Convex} developed by
\cite{Franz2009}. To implement the SHP-algorithm \cite{Loehne_Rudloff2014}
used Benson's linear vector optimisation technique; see \cite{Benson1998},
\cite{Hamel_Lohne_Rudloff2013}. We illustrate the results by a numerical
example computed by means of the primal and dual constructions and compare this
with a similar example presented by \cite{Loehne_Rudloff2014}, who employed
the SHP-algorithm.

We conclude by suggesting a possible extension of the SHP-algorithm to
hedge and price the seller's (short) position in an American option, and
pointing out an inherent difficulty in hedging and pricing the buyer's
(long) position in an American option due to the
essential non-convexity of the problem.

\section{A general duality result}

In this section we present a simple observation that links support functions
with duality in linear vector optimization. The related work of \cite{Luc2011} provides further insight on the connection between support functions and duality. This result will prove useful in comparing the various pricing and hedging algorithms in the following sections.

For a cone $C\subseteq\mathbb{R}^{q}$ we
define a partial ordering $\leq_{C}$ on~$\mathbb{R}^{q}$ by
\[
y\leq_{C}z\iff z-y\in C
\]
and denote by $C^{+}$ the dual (or positive polar) cone of $C$, i.e.
\[
C^{+}=\left\{  x\in\mathbb{R}^{q}:x^{T}y\geq0~\forall y\in C\right\}  .
\]

In what follows we assume that $C$ is a polyhedral cone with non-empty
interior, and there exists some $c\in\interior C$ with $c_{q}=1$.
Suppose that matrices $P\in\mathbb{R}^{q\times d}$ and $B\in\mathbb{R}%
^{m\times d}$ and a vector $b\in\mathbb{R}^{m}$ are given, and consider the
linear vector optimization problem
\begin{equation}
\label{eq:P}\text{minimize }Px \text{ with respect to } \le_{C}\text{ over }
x\in S, \tag{P}%
\end{equation}
with feasible set
\[
S=\{x\in\mathbb{R}^{d}:Bx\ge b\}.
\]
The \emph{upper image} of problem \eqref{eq:P} is the set
\[
\mathcal{P}=P[S] + C.
\]

The dual problem to \eqref{eq:P} is
\begin{equation}
\text{maximize }D^{\ast}(u,w)\text{ with respect to }\leq_{K}\text{ over
}(u,w)\in T, \tag{D$^\ast$}\label{eq:D*}%
\end{equation}
where the linear operator $D^{\ast}:\mathbb{R}^{m}\times\mathbb{R}%
^{q}\rightarrow\mathbb{R}^{q}$ is defined as
\[
D^{\ast}(u,w)=(w_{1},\ldots,w_{q-1},b^{T}u)^{T}\text{ for }(u,w)\in
\mathbb{R}^{m}\times\mathbb{R}^{q},
\]
with $K=\cone\{e^{q}\}$ for $e^{q}=(0,\ldots,0,1)\in\mathbb{R}^{q}$, and with%
\[
T=\{(u,w)\in\mathbb{R}^{m}\times\mathbb{R}^{q}:u\geq0,B^{T}u=P^{T}%
w,c^{T}w=1,w\in C^{+}\}.
\]
The \emph{lower image} of problem \eqref{eq:D*} is the set
\[
\mathcal{D}^{\ast}=D^{\ast}[T]-K.
\]

We now state and prove a general result that links the lower
image~$\mathcal{D}^{\ast}$ of \eqref{eq:D*} with the support function of
$-\mathcal{P}$, where $\mathcal{P}$ is the upper image of \eqref{eq:P}. The
support function $Z:\mathbb{R}^{q}\rightarrow\mathbb{R}$ of $-\mathcal{P}$ is
defined as \cite[see e.g.][p.~28]{rockafellar1996}
\[
Z(x) = \sup\left\{  x^{T}z:z\in-\mathcal{P}\right\}  \text{ for all }%
x\in\mathbb{R}^{q}.
\]
Note that $Z(x)$ is the negative of a scalarization of $\mathcal{P}$ with respect to the weighting vector $x$ \citep[see e.g.][Section~4.1.1]{Lohne2011}. Thus the following result can be regarded as a reformulation of strong geometric duality \citep[see][Theorems~4.40,~4.41]{Lohne2011} by means of the family of scalarizations of~$\mathcal{P}$.

\begin{proposition}
\label{Prop1-ver2}If $C$ contains no lines, i.e.\ if $C\cap(-C)=\{0\}$, then%
\begin{align}
\mathcal{D}^{\ast}  &  =\left\{  w\in\mathbb{R}^{q}:-w_{q}\geq Z\left(
w_{1},\ldots w_{q-1},1-\sum_{i=1}^{q-1}c_{i}w_{i}\right)  \right\}
,\label{Eq:Prop1-1}\\
Z(w)  &  =\left\{
\begin{array}
[c]{ll}%
-\sup\left\{  y\in\mathbb{R}:\frac{1}{c^{T}w}\left(  w_{1},\ldots
,w_{q-1},y\right)  \in\mathcal{D}^{\ast}\right\}  & \text{\upshape if }%
c^{T}w>0,\\
0 & \text{\upshape if }w=0,\\
\infty & \text{\upshape otherwise.}%
\end{array}
\right.  \label{Eq:Prop1-2}%
\end{align}
\end{proposition}

\begin{proof}
 If $C$ contains no lines, then Theorems~4.40 and~4.41 of \cite{Lohne2011} \citep[see also][Remark~3.7]{Hamel_Lohne_Rudloff2013} give
\[
\mathcal{D}^{\ast}=\left\{  w\in\mathbb{R}^{q}:\varphi(y,w)\geq0~\forall
y\in\mathcal{P}\right\},
\]
where the bi-affine coupling function $\varphi:\mathbb{R}%
^{q}\times\mathbb{R}^{q}\rightarrow\mathbb{R}$ is defined as
\[
\varphi(y,w)=\sum_{i=1}^{q-1}y_{i}w_{i}+y_{q}\left(  1-\sum_{i=1}^{q-1}%
c_{i}w_{i}\right)  -w_{q}\text{ for }(y,w)\in\mathbb{R}^{q}\times
\mathbb{R}^{q}.
\]
The function $\varphi$ was first introduced for the special case
$c=(1,\ldots,1)^{T}$ by \cite{Heyde_Lohne2008} and for general $c$ by \cite{Loehne_Rudloff2014}.

Observe that $\varphi(y,w)\geq0$ for all $y\in\mathcal{P}$ if and only if%
\[
-w_{q}\geq\sum_{i-1}^{q-1}y_{i}w_{i}+y_{q}\left(  1-\sum_{i=1}^{q-1}c_{i}%
w_{i}\right)  \quad\text{for all }y\in-\mathcal{P},
\]
that is, if and only if%
\begin{align*}
-w_{q}  &  \geq\sup\left\{  \sum_{i-1}^{q-1}y_{i}w_{i}+y_{q}\left(
1-\sum_{i=1}^{q-1}c_{i}w_{i}\right)  :y\in-\mathcal{P}\right\} \\
&  =Z\left(  w_{1},\ldots w_{q-1},1-\sum_{i=1}^{q-1}c_{i}w_{i}\right)  .
\end{align*}
This proves~\eqref{Eq:Prop1-1}.

Now take any $w\in\mathbb{R}^{d}$ such that $c^{T}w>0$. Then $-y\geq Z(w)$ is
equivalent to $-\frac{y}{c^{T}w}\geq Z\left(  \frac{w}{c^{T}w}\right)  $ since
the support function is positively homogeneous. By~(\ref{Eq:Prop1-1}), the last inequality is in turn equivalent to $\frac
{1}{c^{T}w}\left(  w_{1},\ldots,w_{q-1},y\right)  \in\mathcal{D}^{\ast}$. This
shows that%
\begin{align*}
Z(w)  &  =-\sup\left\{  y\in\mathbb{R}:-y\geq Z(w)\right\} \\
&  =-\sup\left\{  y\in\mathbb{R}:\frac{1}{c^{T}w}\left(  w_{1},\ldots
,w_{q-1},y\right)  \in\mathcal{D}^{\ast}\right\}
\end{align*}
when $c^{T}w>0$. If $w=0$, then $Z(w)=0$ by the definition of the support
function. Finally, take any $w\neq0$ such that $c^{T}w\leq0$. Since
$c\in\interior C$, there is an $\varepsilon>0$ such that $c-\varepsilon w\in
C$. It follows that $\left(  c-\varepsilon w\right)  ^{T}w=c^{T}w-\varepsilon
w^{T}w<0$ because $w^{T}w>0$. As $\mathcal{P}=\mathcal{P}+C$, for any
fixed $x\in\mathcal{P}$ and for each $\lambda>0$ we have $x+\lambda
(c-\varepsilon w)\in\mathcal{P}$. Hence, by the definition of the support
function,
\[
Z(w)\geq-\left(  x+\lambda(c-\varepsilon w)\right)  ^{T}w=-x^{T}%
w-\lambda(c-\varepsilon w)^{T}w
\]
for each $\lambda>0$. Since $(c-\varepsilon w)^{T}w<0$, this means that
$Z(w)=\infty$, completing the proof of~(\ref{Eq:Prop1-2}).
\end{proof}

\begin{remark}
\upshape According to Proposition \ref{Prop1-ver2},
\begin{equation} \label{eq:dstar-as-epigraph}
\mathcal{D}^{\ast}=\left\{  (w_{1},\ldots,w_{q-1},y)\in\mathbb{R}^{q}%
:(w,y)\in-\epi Z,c^{T}w=1\right\}  ,
\end{equation}
so $\mathcal{D}^{\ast}$ can be identified with the section of the cone $-\epi Z$ by the
hyperplane $\{(w,y)\in\mathbb{R}^{q}\times\mathbb{R}:c^{T}w=1\}$
in~$\mathbb{R}^{q+1}$. The convex set~$\mathcal{D}^{\ast}$ (which depends
on~$c$) captures the same information as the support function~$Z$. This is
remarkable given that $Z$ is independent of the arbitrary choice of~$c$. Also note the similarity between \eqref{eq:dstar-as-epigraph} and the representation by \citet[p.~828]{Heyde2013} of the dual image in a more general setting.
\end{remark}

This section concludes with a simple example.

\begin{example}\upshape
\label{ex:1} Suppose that
\[
P=%
\begin{pmatrix}
1 & -1\\
1 & \phantom-1
\end{pmatrix},\
B=%
\begin{pmatrix}
2 & 1\\
1 & 2\\
1 & 0\\
0 & 1
\end{pmatrix},\
b=\begin{pmatrix}6\\6\\0\\0\end{pmatrix},\
C=\cone\left\{  \begin{pmatrix}-3\\\phantom-1\end{pmatrix},\begin{pmatrix}1\\2\end{pmatrix}\right\}  ,
\]
and fix $c=(0,1)^T\in\interior C$. For this data we have
\begin{align*}
\mathcal{P}  &  =\{(z_{1},z_{2})\in\mathbb{R}^{2}:z_{2}\geq\tfrac{1}{3}%
z_{1}+4,z_{2}\geq z_{1},z_{2}\geq-\tfrac{1}{3}z_{1}+4\},\\
\mathcal{D}^{\ast}  &  =\{(w_{1},y)\in\mathbb{R}^{2}:-1\leq w_{1}\leq\tfrac
{1}{3},y\leq4,y-6w_{1}\leq6\}
\end{align*}
\cite[full details in][Example 6.4]{Loehne_Rudloff2011}. The
sets~$\mathcal{P}$ and~$\mathcal{D}^{\ast}$ are represented graphically in
Figure~\ref{fig:1}.
\begin{figure}[ptb]
\begin{center}%
\begin{tabular}
[c]{cc}%
\begin{tikzpicture}[x=0.5/16*\figurewidth,y=0.5/12*\figureheight] \draw[gray!50,fill=gray!50] (8,10) -- (8,8) -- (6,6) -- (0,4) -- (-8,4+8/3) -- (-8,10) -- cycle; \draw[axis,->] (-8,0) -- (8,0) node[below] {$z_1$}; \foreach \x/\xname in {-6,6} \draw (\x,0) node[below] {$\xname$} -- ++(0,1ex); \draw[axis,->] (0,-2) -- (0,10) node[above] {$z_2$}; \foreach \y/\yname in {4,6} \draw (0,\y) node[left] {$\yname$} -- ++(1ex,0); \draw[function] (6,6) -- (8,8); \draw[function] (0,4) -- (6,6); \draw[function] (0,4) -- (-8,4+8/3); \draw (-8,10) node[below right] {$\mathcal{P}$}; \end{tikzpicture} &
\begin{tikzpicture}[x=0.5/2*\figurewidth,y=0.5/12*\figureheight] \draw[gray!50,fill=gray!50] (1/3,-2) -- (1/3,4) -- (-1/3,4) -- (-1,0) -- (-1,-2) -- cycle; \draw[axis,->] (-4/3,0) -- (2/3,0) node[below] {$w_1$}; \foreach \x/\xname in {{-1/3}/-\tfrac{1}{3}\phantom-} \draw (\x,0) node[below] {$\xname$} -- ++(0,1ex); \foreach \x/\xname in {-1/-1\phantom-,{1/3}/\tfrac{1}{3}} \draw (\x,0) node[below right] {$\xname$} -- ++(0,1ex); \draw[axis,->] (0,-2) -- (0,10) node[above] {$y$}; \foreach \y/\yname in {8} \draw (0,\y) node[left] {$\yname$} -- ++(1ex,0); \foreach \y/\yname in {4} \draw (0,\y) node[above left] {$\yname$} -- ++(1ex,0); \draw[function] (-4/3,4) -- (2/3,4); \draw[function] (-4/3,-2) -- (2/3,10); \draw[function,] (-1,-2) -- (-1,10); \draw[function] (1/3,-2) -- (1/3,10); \draw (1/3,-2) node[above left] {$\mathcal{D}^\ast$}; \end{tikzpicture}
\end{tabular}
\end{center}
\caption{Upper and lower images in Example \ref{ex:1}}%
\label{fig:1}%
\end{figure}
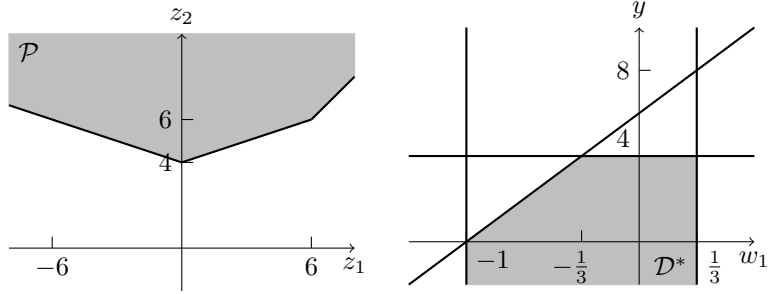

The support function $Z$ is finite on its effective domain, which consists of vectors $w\in\mathbb{R}^2$ such that $x^Tw\le 0$ for each $x\in-\mathcal{P}$, so
\[
\dom Z = \{w\in\mathbb{R}^{2}:Z(w)<\infty\} = \{(w_{1},w_{2})\in
\mathbb{R}^{2}:w_{2}\ge-w_{1},w_{2}\ge3w_{1}\}.
\]
For each $w\in\dom Z$ the linear function $x\mapsto x^Tw$ takes a maximum at one of the extreme points $(0,-4),(-6,-6)$ of the convex set $-\mathcal{P}$, hence
\[
Z(w) =\sup\{x^Tw:x\in-\mathcal{P}\} = \max\{-4w_{2},-6w_{1}-6w_{2}\}.
\]
This means that
\begin{multline*}
\{(w_{1},y)\in\mathbb{R}^{2}:(w,y)\in-\epi Z,c^{T}w=1\}\\
\begin{aligned} &=\{(w_1,y)\in\mathbb{R}^2:y\le-Z(w_1,w_2),(w_1,w_2)\in\dom Z,w_2 = 1\}\\ &=\{(w_1,y)\in\mathbb{R}^2:y\le-Z(w_1,1),-1\le w_1 \le \tfrac{1}{3}\}\\ &=\{(w_1,y)\in\mathbb{R}^2:y\le 4, y\le 6w_1 + 6,-1\le w_1 \le \tfrac{1}{3}\}=\mathcal{D}^\ast. \end{aligned}
\end{multline*}
This identifies $\mathcal{D}^{\ast}$ with the section of $-\epi Z$ by the hyperplane
\[
\{(w,y)\in\mathbb{R}^{2}\times\mathbb{R}:c^Tw=1\}=\{(w_1,w_2,y)\in\mathbb{R}^{3}:w_2=1\}.
\]
\end{example}

\section{Pricing and hedging European options under proportional transaction
costs}

\subsection{Currency model}\label{Sect:curr-mod}

The model is based on a filtered probability space $(\Omega,\mathcal{F}%
,\mathbb{P};(\mathcal{F}_{t})_{t=0}^{T})$. We assume that $\Omega$ is finite,
that $\mathcal{F}_{0}=\{\emptyset,\Omega\}$, that $\mathcal{F}_{T}%
=\mathcal{F}=2^{\Omega}$ and that $\mathbb{P}(\omega)>0$ for all $\omega
\in\Omega$. For each $t$ denote by $\Omega_{t}$ the collection of atoms of
$\mathcal{F}_{t}$, called the time~$t$ \emph{nodes} of the associated stock price tree
model. Note that $\Omega_{0}=\{\Omega\}$ and $\Omega_{T}=\{\{w\}:\omega
\in\Omega\}$. For every $t<T$ a node $\nu\in\Omega_{t+1}$ is said to be a \emph{successor} of
a node $\mu\in\Omega_{t}$ if $\nu\subseteq\mu$. We denote for all $\mu\in\Omega_t$
\[
 \successors \mu = \{\nu\in\Omega_{t+1}:\nu\text{ a successor of }\mu\}.
\]

For each $t$ let $\mathcal{L}_{t}=\mathcal{L}^{0}(\mathbb{R}^{d}%
;\mathcal{F}_{t})$ be the collection of $\mathcal{F}_{t}$-measurable
$\mathbb{R}^{d}$-valued random variables. We identify elements of
$\mathcal{L}_{t}$ with functions on $\Omega_{t}$ whenever convenient.

We consider the discrete-time currency model introduced by \cite{kabanov1999}
and studied by others. The model contains $d$ assets or currencies. At each
trading date $t=0,1,\ldots,T$ one unit of each asset ${k}=1,\ldots,d$ can be
obtained by exchanging $\pi^{{jk}}_{t}>0$ units of asset ${j}=1,\ldots,d$. We
assume that {the exchange rates $\pi^{{jk}}_{t}$ are $\mathcal{F}_{t}%
$-measurable and} $\pi^{{jj}}_{t}=1$ for all $t$ and ${j,k}$.

We say that a portfolio $x\in\mathcal{L}_{t}$ can be \emph{exchanged} into
a portfolio $y\in\mathcal{L}_{t}$ at time~$t$ whenever there are
$\mathcal{F}_{t}$-measurable random variables $\beta^{jk}\geq0$,
$j,k=1,\ldots,d$ such that for all $k=1,\ldots,d$
\[
y^{k}=x^{k}+\sum_{j=1}^{d}\beta^{jk}-\sum_{j=1}^{d}\beta^{kj}\pi_{t}^{kj},
\]
where $\beta^{jk}$ represents the number of units of asset~$k$ received as a
result of exchanging some units of asset~$j$.

The \emph{solvency cone} $\mathcal{K}_{t}\subseteq\mathcal{L}_{t}$ is the set
of portfolios that are \emph{solvent} at time~$t$, i.e.\ those portfolios at
time $t$ that can be exchanged into portfolios with non-negative holdings in
all~$d$ assets. It is straightforward to show that $\mathcal{K}_{t}$ is the
convex cone generated by the canonical basis $e^{1},\ldots,e^{d}$ of
$\mathbb{R}^{d}$ and the vectors $\pi^{{jk}}_{t}e^{j}-e^{k}$ for
${j,k}=1,\ldots,d$, and so $\mathcal{K}_{t}$ is a polyhedral cone. Note that
$\mathcal{K}_{t}$ contains all the {non-negative} elements of $\mathcal{L}%
_{t}$.

A \emph{self-financing strategy} $y=(y_{t})_{t=0}^{T}$ is a predictable
$\mathbb{R}^{d}$-valued process (i.e.\ $y_{0}\in\mathcal{L}_{0}$ and $y_{t}%
\in\mathcal{L}_{t-1}$ for $t=1,\ldots,T$) such that%
\[
y_{t}-y_{t+1}\in\mathcal{K}_{t}\quad\text{for all }t=0,\ldots,T-1
\]
Here $y_{0}\in\mathcal{L}_{0}$ is the initial endowment, and $y_{t}%
\in\mathcal{L}_{t-1}$ for each $t=1,\ldots,T$ is the portfolio held from time
$t-1$ to time $t$. Let~$\Phi$ be the set of self-financing strategies.

A self-financing strategy $y=(y_{t})\in\Phi$ is called an \emph{arbitrage
opportunity} if $y_{0}=0$ and {there is a portfolio $x\in\mathcal{L}%
_{T}\setminus\{0\}$ with non-negative holdings in all~$d$ assets such that
$y_{T}-x\in\mathcal{K}_{T}$.} This notion of arbitrage was considered by
\cite{schachermayer2004}, and its absence is formally different but
equivalent to the weak no-arbitrage condition introduced by \cite{kabanov_stricker2001b}.

\begin{theorem}
[\cite{kabanov_stricker2001b,schachermayer2004}]\label{th:2012-10-03:ftap}
The model admits no arbitrage opportunity if and only if there exists a probability
measure~$\mathbb{Q}$ equivalent to~$\mathbb{P}$ and an $\mathbb{R}^{d}$-valued
$\mathbb{Q}$-martingale $S=(S_{t})$ such that
\begin{equation}
S_{t}\in\mathcal{K}_{t}^{+}\setminus\{0\}\text{ for all }t,
\label{eq:th:2012-10-03:ftap}%
\end{equation}
where $\mathcal{K}_{t}^{+}$ is the dual cone of~$\mathcal{K}_{t}$.
\end{theorem}

\begin{remark}\upshape
A pair $(\mathbb{Q},S)$ satisfying the conditions in Theorem
\ref{th:2012-10-03:ftap} is called a \emph{consistent
pricing pair}. In place of such a pair $(\mathbb{Q},S)$ one can equivalently use the so-called
\emph{consistent price process} $S_{t}\mathbb{E}%
_{\mathbb{P}}(\frac{d\mathbb{Q}}{d\mathbb{P}}|\mathcal{F}_{t})$; see \cite{schachermayer2004}.
\end{remark}

\subsection{European options}

A \emph{European option} with expiry time $T>0$ and payoff $\xi\in
\mathcal{L}_{T}$ is a contract that gives its holder (i.e.\ the option buyer)
the right to receive a portfolio $\xi$ of currencies at time~$T$. On the other
hand, the writer (seller) of the option is obliged to deliver this portfolio
to the buyer.

To hedge against this liability the writer can follow a self-financing
strategy $y\in\Phi$ such that $y_{T}-\xi\in\mathcal{K}_{T}$. The initial
endowment $y_{0}$ of such a strategy~$y$ is called a \emph{superhedging
portfolio}, and the strategy $y$ itself is called a \emph{superhedging
strategy} for the European option~$\xi$.

The \emph{ask price} (\emph{seller's price}, \emph{superhedging price})
$\pi_{i}^{a}(\xi)$ of the European option in currency $i=1,\ldots,d$ can be
understood as the lowest value $x$ such that the portfolio consisting of $x$
units of currency~$i$ and no other currency is a superhedging portfolio
for~$\xi$. In other words,%
\[
\pi_{i}^{a}(\xi)=\min\left\{  x\in\mathbb{R}:xe^{i}\text{ is a superhedging
portfolio for }\xi\right\}  .
\]

On the other hand, to hedge his position the option buyer would like to follow
a self-financing strategy $y\in\Phi$ such that $y_{T}+\xi\in\mathcal{K}_{T}$.
Here $-y_{0}$ is a portfolio of currencies which the option buyer could borrow
at time~$0$ and would be able to settle later by following the strategy~$y$
and using the payoff~$\xi$ to be received on exercising the option at
time~$T$. We call $-y_{0}$ a \emph{subhedging portfolio} and $-y$ a
\emph{subhedging strategy} for the European option~$\xi$.

The \emph{bid price} (\emph{buyer's price}, \emph{subhedging price}) $\pi
_{i}^{b}(\xi)$ of the European option in currency $i=1,\ldots,d$ can be
understood as the highest value $x$ such that the portfolio consisting of $x$
units of currency~$i$ and no other currency is a subhedging portfolio
for~$\xi$,%
\[
\pi_{i}^{b}(\xi)=\max\left\{  x\in\mathbb{R}:xe^{i}\text{ is a subhedging
portfolio for }\xi\right\}  .
\]
It is the highest amount in currency~$i$ that an option holder could raise by
using the option as collateral.

Observe that $-y$ is a subhedging strategy for a European option $\xi$ if and
only if $y$ is a superhedging strategy for~$-\xi$. It follows immediately that%
\[
\pi_{i}^{b}(\xi)=-\pi_{i}^{a}(-\xi).
\]
Because of these relationships it is sufficient to develop algorithms for
hedging and pricing the seller's (short) position in a European option.%

\subsection{Primal and dual constructions\label{Sect:primal-dual-constr}}

The constructions presented here for European options are a special case of
those developed by \cite{Roux_Zastawniak2014} to hedge and price the much
wider class of American type options under proportional transaction costs.
Construction~4.2 in \cite{Roux_Zastawniak2014}, which produces the set of
superhedging portfolios, takes a particularly simple form in this special case:

\begin{itemize}
\item For each $\omega\in\Omega_{T}$ put%
\[
\mathcal{Z}_{T}^{\omega}=\xi^{\omega}+\mathcal{K}_{T}^{\omega}.
\]

\item If $\mathcal{Z}_{t+1}$ has already been constructed for some
$t=0,1,\ldots,T-1$, then for each $\omega\in\Omega_{t}$ put%
\begin{align*}
\mathcal{W}_{t}^{\omega}  &  =\bigcap_{\omega^{\prime}\in\successors\omega
}\mathcal{Z}_{t+1}^{\omega^{\prime}},\\
\mathcal{Z}_{t}^{\omega}  &  =\mathcal{W}_{t}^{\omega}+\mathcal{K}_{t}%
^{\omega}%
\end{align*}
(To link this with Construction~4.2 in \cite{Roux_Zastawniak2014} observe that
the formula for $\mathcal{W}_{t}$ can be written concisely as $\mathcal{W}%
_{t}=\mathcal{Z}_{t+1}\cap\mathcal{L}_{t}$.)
\end{itemize}

\noindent For each $t$ the set $\mathcal{Z}_{t}$ consists of all portfolios
that allow the seller to hedge the option by following a self-financing
strategy between times $t$ and~$T$. In particular, $\mathcal{Z}_{0}$ is the
set of superhedging portfolios. The ask price of the option can be expressed
in terms of $\mathcal{Z}_{0}$ as%
\begin{equation}
\pi_{i}^{a}(\xi)=\min\left\{  x\in\mathbb{R}:xe^{i}\in\mathcal{Z}_{0}\right\}
. \label{Eq:pi-def}%
\end{equation}

The above construction involves two standard operations on polyhedral convex
sets, namely the intersection of finitely many such sets and the algebraic sum
of such a set and a polyhedral convex cone. Both operations can be implemented
using standard geometric methods in existing software libraries, for example, \emph{Parma Polyhedra Library} \citep{Bagnara_Hill_Zaffanella2008a} and \emph{PolyLib} \citep[among others]{LeVerge1992,Wilde1993,IRISA2001,Loechner2010}.
As soon as the set~$\mathcal{Z}_{0}$ of superhedging portfolios has been computed in this manner, it becomes a routine task to evaluate the option price $\pi
_{i}^{a}(\xi)$ using~(\ref{Eq:pi-def}).
\cite{Roux_Zastawniak2014}
provided a numerical implementation of this procedure for hedging and pricing
European options (and much more generally, American type options) in currency
markets with transaction costs by using the \emph{Maple} package \emph{Convex} \citep{Franz2009}.

Moreover, once the $\mathcal{Z}_{t}$ have been constructed, it is
straightforward to compute a superhedging strategy starting from any
superhedging portfolio $y_{0}\in\mathcal{Z}_{0}$. Namely, if $y_{t}%
\in\mathcal{Z}_{t}$ has already been computed for some $t=0,1,\ldots,T-1$, we
can take $y_{t+1}\in\left(  y_{t}-\mathcal{K}_{t}\right)  \cap\mathcal{W}_{t}%
$. The intersection is non-empty since $\mathcal{Z}_{t}=\mathcal{W}%
_{t}+\mathcal{K}_{t}$, so it is always possible to find such $y_{t+1}$, though
it may be non-unique. The self-financing condition $y_{t}-y_{t+1}%
\in\mathcal{K}_{t}$ is clearly satisfied. Moreover, since $\mathcal{W}%
_{t}=\mathcal{Z}_{t+1}\cap\mathcal{L}_{t}$, it follows that $y_{t+1}$ is
$\mathcal{F}_{t}$-measurable, so $y$ constructed in this manner will be a
predictable process. It also follows that $y_{t+1}\in\mathcal{Z}_{t+1}$, which
makes it possible to iterate the procedure.

It is also possible to follow the construction using convex dual objects to
the~$\mathcal{Z}_{t}$. We introduce the support functions
\[
Z_{t}(x)=\sup\left\{  x^{T}z:z\in-\mathcal{Z}_{t}\right\}  ,\quad
W_{t}(x)=\sup\left\{  x^{T}z:z\in-\mathcal{W}_{t}\right\}
\]
and the linear function%
\[
U(x)=-x^{T}\xi
\]
defined for all $x\in\mathbb{R}^{d}$. If we need to make the dependence on $\omega\in\Omega$ explicit in these functions, we shall write $Z^\omega_t,W^\omega_t,U^\omega$. The above construction (we call it the
\emph{primal construction}) can now be written in the following equivalent
form (called the \emph{dual construction}); see Lemma~5.5 in
\cite{Roux_Zastawniak2014}:

\begin{itemize}
\item For each $\omega\in\Omega_{T}$%
\[
Z_{T}^{\omega}=\left\{
\begin{array}
[c]{ll}%
U^{\omega} & \text{on }\mathcal{K}_{T}^{+\omega},\\
\infty & \text{otherwise.}%
\end{array}
\right.
\]
This is the linear function $U^{\omega}$ restricted to the domain
$\mathcal{K}_{T}^{+\omega}$.

\item Suppose that $Z_{t+1}$ has been constructed for some $t=0,1,\ldots,T-1$.
Then, for each node $\omega\in\Omega_{t}$ let $W_{t}^{\omega}$ be the convex
hull of the family of convex functions $Z_{t+1}^{\omega^{\prime}}$ indexed by
$\omega^{\prime}\in\successors\omega$, and let $Z_{t}^{\omega}$ be the
restriction of $W_{t}^{\omega}$ to the domain $\mathcal{K}_{t}^{+\omega}$:%
\begin{align*}
W_{t}^{\omega}  &  =\conv\left\{  Z_{t+1}^{\omega^{\prime}}:\omega^{\prime}%
\in\successors\omega\right\}  ,\\
Z_{t}^{\omega}  &  =\left\{
\begin{array}
[c]{ll}%
W_{t}^{\omega} & \text{on }\mathcal{K}_{t}^{+\omega},\\
\infty & \text{otherwise.}%
\end{array}
\right.
\end{align*}

\end{itemize}

\noindent Once $Z_{0}$ has been computed, the ask price of the option can be
obtained as (see Theorem~4.4 in \cite{Roux_Zastawniak2014})%
\[
\pi_{i}^{a}(\xi)=-\min\left\{  Z_{0}(x):x\in\mathbb{R}^{d},x_{i}=1\right\}  .
\]

This dual construction also lends itself well to computer implementation.
Taking the convex hull of finitely many polyhedral convex functions and
restricting the domain of such a function to a given polyhedral convex cone
are operations equivalent to some standard operations on polyhedral convex sets, which are widely available in computer packages such as the \emph{Convex} library in \emph{Maple} used by \cite{Roux_Zastawniak2014}.

Observe that the dual construction, which follows from Lemma~5.5 in \cite{Roux_Zastawniak2014} specialised to the case of European options, is equivalent to the construction in Corollary~6.3 of \cite{Loehne_Rudloff2014}. The only difference is that the dual construction is expressed in terms of the support functions $Z_t$ and~$W_t$, whereas \cite{Loehne_Rudloff2014} use $\widetilde V_t(x)=-Z_t(x)$ and $V_t(x)=-W_t(x)$ defined for all $x$'s on the hyperplane in $\mathbb{R}^d$ given by the condition $x^i=1$. Both are a straightforward extension to $d$ assets of the construction stated in Algorithm~4.1 of \cite{Roux_Tokarz_Zastawniak2008} in the case of $2$ assets.

\subsection{SHP-algorithm\label{Sect:SHP-alg}}

\cite{Loehne_Rudloff2014} consider the same problem of pricing and hedging
European options (though not options of American type). In particular, the
same sets as in the primal construction above are denoted by
\cite{Loehne_Rudloff2014} as%
\[
SHP_{t}(\xi)=\mathcal{Z}_{t}.
\]
These authors propose a different construction of the $\mathcal{Z}_{t}$ based
on linear vector optimisation methods and geometric duality.

From this perspective, $S=\mathcal{W}_{t}$ can be viewed as the feasible set
of a linear vector optimisation problem~(\ref{eq:P}). If the solvency cone
$\mathcal{K}_{t}$ contains no lines, which means that there are non-zero
transaction costs between any two currencies, then the matrix $P$
in~(\ref{eq:P}) is just the $d\times d$ unit matrix, and the ordering cone is
$C=\mathcal{K}_{t}$. The upper image of the linear vector optimisation
problem~(\ref{eq:P}) is%
\[
\mathcal{P}=P[S]+C=\mathcal{W}_{t}+\mathcal{K}_{t}=\mathcal{Z}_{t}.
\]
Because $C$ contains no lines, Benson's algorithm, see \cite{Benson1998} or
\cite{Hamel_Lohne_Rudloff2013}, can be applied to compute a solution to the
dual problem~(\ref{eq:D*}) and hence the corresponding lower
image~$\mathcal{D}^{\ast}$. The Benson algorithm yields simultaneously a
solution to (\ref{eq:P}) and gives the upper image~$\mathcal{P}=\mathcal{Z}%
_{t}$. We know from Proposition~\ref{Prop1-ver2} that if $C$ contains no
lines, then $\mathcal{D}^{\ast}$ can be identified with a section of the
epigraph of the support function $Z$ of $-\mathcal{P}$. Since $\mathcal{P=Z}%
_{t}$, it follows that $Z=Z_{t}$ is the function from the dual construction in Section~\ref{Sect:primal-dual-constr}.

A complication arises when the solvency cone $\mathcal{K}_{t}$ contains some
lines, which means that there are currencies which can be exchanged into one
another without incurring any transaction costs. This is dealt with by taking
$P$ to be the matrix representing the so-called liquidation map, a linear map
which amounts to liquidating all but one of the assets that can be exchanged
into one another without transaction costs; see (4.1) in
\cite{Loehne_Rudloff2014} for the precise definition of~$P$. In this case
$C=P[\mathcal{K}_{t}]$ contains no lines because there are no longer any
assets that can be exchanged into one another without transaction costs. Then
the upper image of the linear vector optimisation problem~(\ref{eq:P}) is%
\[
\mathcal{P}=P[S]+C=P[\mathcal{W}_{t}+\mathcal{K}_{t}]=P[\mathcal{Z}_{t}].
\]
Since $C$ contains no lines, Benson's algorithm can also be applied in this
case to compute a solution to the dual problem~(\ref{eq:D*}) and hence the
corresponding lower image~$\mathcal{D}^{\ast}$. The Benson algorithm yields
simultaneously a solution to (\ref{eq:P}) and gives the upper
image~$\mathcal{P}=P[\mathcal{Z}_{t}]$. This then gives $\mathcal{Z}%
_{t}=\{x\in\mathcal{L}_{t}:Px\in\mathcal{P}\}$ as the inverse image
of~$\mathcal{P}$ under $P$. Once again by Proposition~\ref{Prop1-ver2}, since
$C$ contains no lines, it follows that $\mathcal{D}^{\ast}$ can be identified
with a section of the epigraph of the support function $Z$ of $-\mathcal{P}%
=-P[\mathcal{Z}_{t}]$. This is related to $Z_{t}$, the support function
of~$-\mathcal{Z}_{t}$, by $Z(x)=Z_{t}(P^{T}x)$.

\section{Example}\label{sec:num-example}

In this section we present {an} example to illustrate the numerical procedures discussed in Section~\ref{Sect:primal-dual-constr}. Consider a model involving three assets, with time horizon $\tau=1$ and with $T=4$ time steps. Two of the assets are risky with correlated returns, and follow the two-asset {recombinant} \citet{Korn_Muller2009} model with Cholesky decomposition. That is, there are $(t+1)^2$ possibilities for the stock prices $S_t=(S^1,S^2)$ at each time step $t=0,\ldots,T$, indexed by pairs $(j_1,j_2)$ where $1\le j_1,j_2\le t+1$, and each non-terminal node with stock price $S_t(j_1,j_2)$ has four successors, associated with the stock prices $S_{t+1}(j_1,j_2)$, $S_{t+1}(j_1+1,j_2)$, $S_{t+1}(j_1,j_2+1)$ and $S_{t+1}(j_1+1,j_2+1)$.  With $\Delta=\tfrac{\tau}{T}$ defined for convenience, the stock prices are given by
\begin{align*}
 S^1_t(j_1,j_2) &= S^1_0e^{\left(r-\tfrac{1}{2}\sigma_1^2\right)t\Delta +(2j_1-t-2)\sigma_1\sqrt{\Delta}},\\
 S^2_t(j_1,j_2) &= S^2_0e^{\left(r-\tfrac{1}{2}\sigma_2^2\right)t\Delta +\left((2j_1-t-2)\rho + (2j_2-t-2)\sqrt{1-\rho^2}\right)\sigma_2\sqrt{\Delta}}
\end{align*}
for $t=0,\ldots,T$ and $j_1,j_2=1,\ldots,t+1$, where $S^1_0=45$ and $S^2_0=50$ are the initial {stock prices}, $\sigma_1=15\%$ and $\sigma_2=20\%$ are the volatilities of {the} returns and $\rho=20\%$ is the correlation between {the log returns on the two stocks}. The third asset is a risk-free bond with nominal interest rate $r=5\%$ and value process
\[B_t=(1+r\Delta)^{-(T-t)}\text{ for }t=0,\ldots,T.\]
Proportional transaction costs are introduced by allowing the asset prices to have constant (proportional) bid-ask spreads, i.e.~the bid and ask prices are
\begin{align*}
 S^{1b}_t &=(1-k_1)S^1_t, & S^{1a}_t &=(1+k_1)S^1_t, \\
 S^{2b}_t &=(1-k_2)S^2_t, & S^{2a}_t &=(1+k_2)S^2_t, \\
 B^b_t &=(1-k_3)B_t, & B^a_t &=(1+k_3)B_t
\end{align*}
for $t=0,\ldots,T$, where $k_1=2\%$, $k_2=4\%$ and $k_3=1\%$. The matrix of exchange rates at each time step $t$ is then
\[
 \begin{pmatrix}\pi^{11}_t&\pi^{12}_t&\pi^{13}_t\\ \pi^{21}_t&\pi^{22}_t&\pi^{23}_t\\ \pi^{31}_t&\pi^{32}_t&\pi^{33}_t\end{pmatrix}
 =\begin{pmatrix}1&\frac{S^{2a}_t}{S^{1b}_t}&\frac{B^{a}_t}{S^{1b}_t}\\ \frac{S^{1a}_t}{S^{2b}_t}&1&\frac{B^{a}_t}{S^{2b}_t}\\ \frac{S^{1a}_t}{B^b_t}&\frac{S^{2a}_t}{B^b_t}&1\end{pmatrix},
\]
and the solvency cone is
\[
 \mathcal{K}_t=\cone\left\{\begin{pmatrix}\phantom-S^{2a}_t\\-S^{1b}_t\\\phantom-0\end{pmatrix},
  \begin{pmatrix}\phantom-B^{a}_t\\\phantom-0\\-S^{1b}_t\end{pmatrix},
  \begin{pmatrix}-S^{2b}_t\\\phantom-S^{1a}_t\\\phantom-0\end{pmatrix},
  \begin{pmatrix}\phantom-0\\\phantom-B^{a}_t\\-S^{1b}_t\end{pmatrix},
  \begin{pmatrix}-B^b_t\\\phantom-0\\\phantom-S^{1a}_t\end{pmatrix},
  \begin{pmatrix}\phantom-0\\-B^b_t\\\phantom-S^{2a}_t\end{pmatrix}\right\}.
\]
This model was also considered by \citet[Section 5.2]{Loehne_Rudloff2014}; note that the assets have been reordered in the present paper.

Consider an exchange option with physical delivery and payoff
\[
 \xi=(\mathbf{1}_{\{S^{1a}_T\ge S^{2a}_T\}},-\mathbf{1}_{\{S^{1a}_T\ge S^{2a}_T\}},0)
\]
that matures at time step $T$. \citet[Example 5.3]{Loehne_Rudloff2014} reported
\[
  SHP_0=\conv\left\{\left(\begin{array}{d} 0.584\\-0.260\\-7.760\end{array}\right),
    \left(\begin{array}{d}0.498\\-0.331\\0.000\end{array}\right),
    \left(\begin{array}{d}0.347\\-0.446\\13.341\end{array}\right)\right\} + \mathcal{K}_0,
\]
and gave the ask price of the exchange option in terms of the bond as
\[
\pi_3^{a}(\xi) = 7.418.
\]
The {boundary} of $SHP_0$ is depicted in Figure~\ref{fig:2}. Application of the primal construction in Section \ref{Sect:primal-dual-constr} produces\begin{figure}
\begin{center}%
\begin{tabular}
[c]{cc}%
\includegraphics[width=0.5\figurewidth]{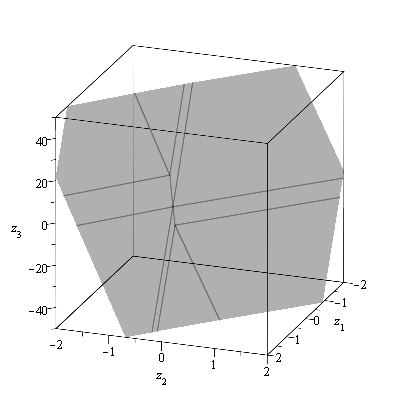} &
\includegraphics[width=0.5\figurewidth]{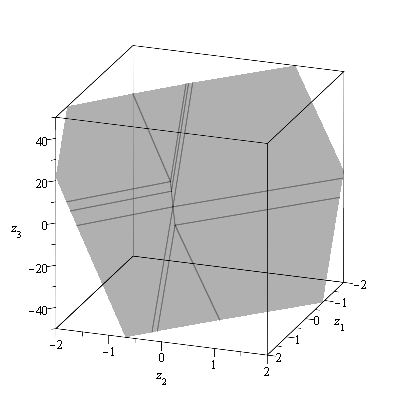} \\
$SHP_0$ & $\mathcal{Z}_0$
\end{tabular}
\end{center}
\caption{{Boundary of the set} of superhedging endowments}%
\label{fig:2}%
\end{figure}
\[
  \mathcal{Z}_0=\conv\left\{\left(\begin{array}{d}  0.584\\-0.260\\-7.760\end{array}\right),
  \left(\begin{array}{d} 0.498\\-0.331\\0.000\end{array}\right),
   \left(\begin{array}{d} 0.399\\-0.406\\8.714\end{array}\right),
   \left(\begin{array}{d} 0.424\\-0.388\\6.564\end{array}\right)\right\} + \mathcal{K}_0,
\]
from which the ask price of the exchange option in terms of each of three assets can be computed as
\begin{align*}
\pi_1^{a}(\xi) &= 0.152, &\pi_2^{a}(\xi) &= 0.146, & \pi_3^{a}(\xi) &= 7.418.
\end{align*}
There is substantial agreement between $SHP_0$ and $\mathcal{Z}_0$, which can be confirmed visually (see Figure~\ref{fig:2}), and in view of the agreement on the ask price $\pi_3^a(\xi)$, we ascribe the differences in the specifications of $SHP_0$ and $\mathcal{Z}_0$ to the error level chosen in Benson's algorithm. Finally, application of the dual construction in Section \ref{Sect:primal-dual-constr} produces the support function $Z_0$ of $-\mathcal{Z}_0$. The set
\[
  \mathcal{D}^\ast_0=\{(w_1,w_2,y):y\le-Z_0(w_1,w_2,1)\}\]
is the lower image of the dual problem~(\ref{eq:D*}) with the choice $c=(0,0,1)^T$.
It has 12 vertices
\begin{gather*}
\left(\begin{array}{d}48.726\\51.930\\7.081\end{array}\right),
  \left(\begin{array}{d}48.726\\51.681\\7.178\end{array}\right),
  \left(\begin{array}{d}45.888\\54.050\\4.981\end{array}\right),
  \left(\begin{array}{d}48.726\\55.201\\5.702\end{array}\right),
  \left(\begin{array}{d}45.888\\49.946\\6.048\end{array}\right),
  \left(\begin{array}{d}48.726\\50.955\\7.418\end{array}\right), \\
  \left(\begin{array}{d}48.573\\50.796\\7.395\end{array}\right),
  \left(\begin{array}{d}47.761\\49.946\\7.141\end{array}\right),
  \left(\begin{array}{d}46.565\\54.907\\5.012\end{array}\right),
  \left(\begin{array}{d}46.815\\55.201\\4.982\end{array}\right),
  \left(\begin{array}{d}46.405\\54.718\\5.018\end{array}\right),
  \left(\begin{array}{d}45.888\\54.108\\4.962\end{array}\right),
\end{gather*}
{and is depicted in Figure~\ref{fig:3}.}%
\begin{figure}
\begin{center}%
\includegraphics[width=0.5\figurewidth]{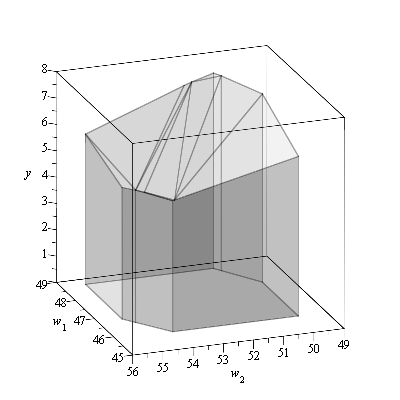}
\end{center}
\caption{Lower image $\mathcal{D}^\ast_0$ {associated with $Z_0$}}%
\label{fig:3}
\end{figure}
The {maximum} of $\mathcal{D}^\ast_0$ in the $y$-direction is
\[
\pi_3^{a}(\xi)=7.418.
\]%

We conclude this numerical example by demonstrating the procedure of finding a superhedging {strategy} $y=(y_t)_{t=0}^T$ starting from the initial endowment \[y_0=(0,0,\pi_3^{a}(\xi))^T\in\mathcal{Z}_0\] along the price path in Table \ref{tab:1}. At each time step $t$ the portfolio $y_t$ (indicated by a dot on the graph of the {boundary} of $\mathcal{Z}_t$ in Table \ref{tab:1}) is {rebalanced} into a portfolio
\[y_{t+1}\in(y_t-\mathcal{K}_t)\cap\mathcal{W}_t\subseteq\mathcal{Z}_{t+1}.\]
As can be seen in Table \ref{tab:1}, for this particular path the set $(y_t-\mathcal{K}_t)\cap\mathcal{W}_t$ is a singleton at time steps $t=0$ and $t=1$, which means that there is only one choice for $y_{t+1}$. At time steps $t=2$ and {$t=3$ this} set is a convex polytope, and the choice of $y_{t+1}$ is no longer unique, which means that other considerations (e.g.~a preference for holding one asset over another, or a preference not to trade) may be used to select $y_{t+1}$ {in $(y_t-\mathcal{K}_t)\cap\mathcal{W}_t$}. In this demonstration we adopted a minimum-trading rule, that is, whenever possible we selected $y_{t+1}=y_t$. {At the final time step $t=4$ we have
\[
y_4-\xi = \left(\begin{array}{d}0.641\\-0.491\\0.000\end{array}\right)
- \left(\begin{array}{d}1.000\\-1.000\\0.000\end{array}\right)
= \left(\begin{array}{d}-0.359\\0.509\\0.000\end{array}\right)
\in\mathcal{K}_4.
\]}
\begin{table}
\begin{center}
\begin{tabular}{|cc@{}c@{}c@{}c@{}|}
\hline
$t$ & $(j_1,j_2)$ & $y_t$ & $\mathcal{Z}_t$ & $(y_t-\mathcal{K}_t)\cap\mathcal{W}_t$ \\
\hline
0 & (1,1) & $\left(\begin{array}{d}0.000\\0.000\\7.418\end{array}\right)$ & \raisebox{-0.5\height}{\includegraphics[height=0.185\textheight]{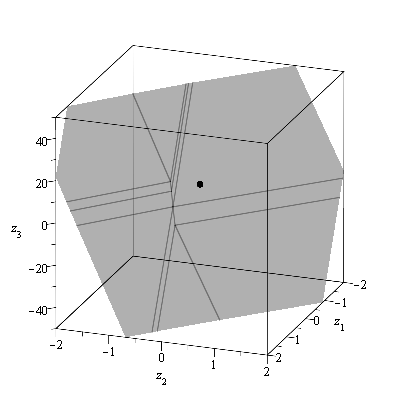}} & $\left\{\left(\begin{array}{d}0.498\\ -0.331\\ 0.000\end{array}\right)\right\}$ \\
1 & (2,1) & $\left(\begin{array}{d}0.498\\ -0.331\\ 0.000\end{array}\right)$ & \raisebox{-0.5\height}{\includegraphics[height=0.185\textheight]{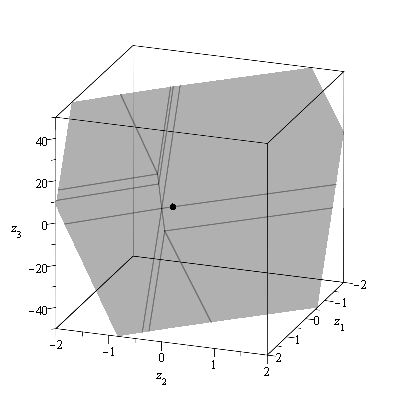}} & $\left\{\left(\begin{array}{d}0.641\\-0.491\\0.000\end{array}\right)\right\}$\\
2 & (2,1) & $\left(\begin{array}{d}0.641\\-0.491\\0.000\end{array}\right)$ & \raisebox{-0.5\height}{\includegraphics[height=0.185\textheight]{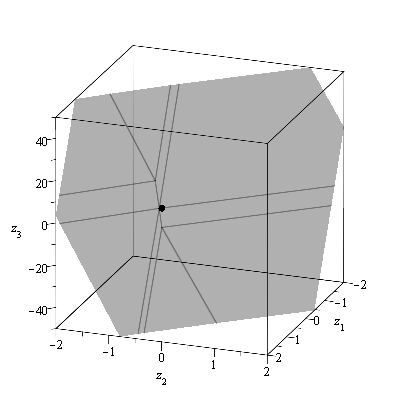}} & \raisebox{-0.5\height}{\includegraphics[height=0.185\textheight]{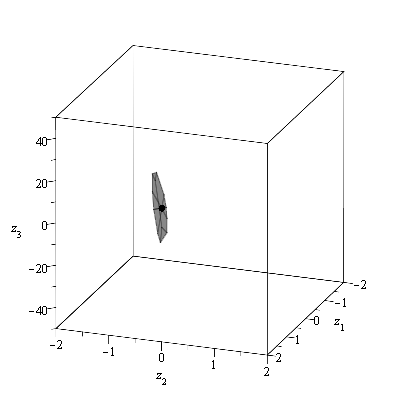}} \\
3 & (3,2) & $\left(\begin{array}{d}0.641\\-0.491\\0.000\end{array}\right)$ & \raisebox{-0.5\height}{\includegraphics[height=0.185\textheight]{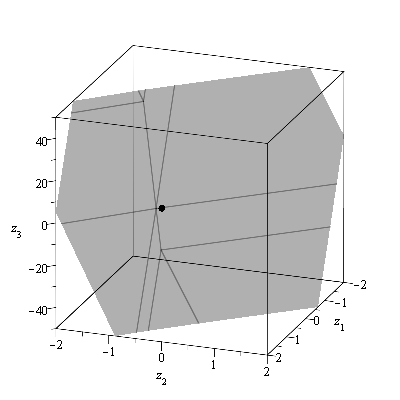}} & \raisebox{-0.5\height}{\includegraphics[height=0.185\textheight]{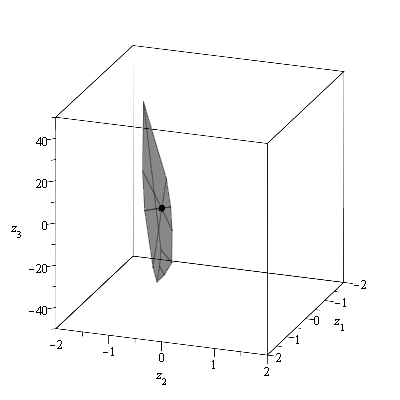}} \\
4 & (3,2) & $\left(\begin{array}{d}0.641\\-0.491\\0.000\end{array}\right)$ & \raisebox{-0.5\height}{\includegraphics[height=0.185\textheight]{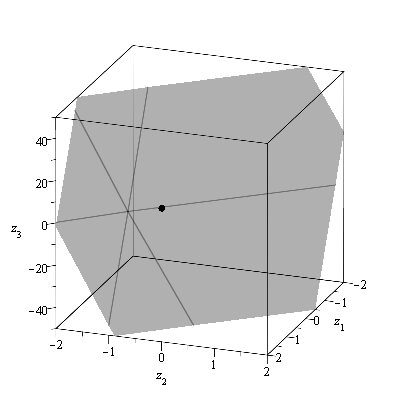}} & N/A \\
\hline
\end{tabular}
\caption{Superhedging {strategy} along a path}
\label{tab:1}
\end{center}
\end{table}

\section{Representation of superhedging price}

In this section we briefly present and compare the result of
\cite{Loehne_Rudloff2014} and \cite{Roux_Zastawniak2014} concerning the
representation of the superhedging price of a European option in terms of risk-neutral expectations of the payoff~$\xi$:%
\begin{equation}
\pi_{i}^{a}(\xi)=\sup_{(\mathbb{Q},S)\in\mathcal{P}^{i}}\mathbb{E}%
_{\mathbb{Q}}((\xi^{T}S_{T})), \label{Eq:price-repr}%
\end{equation}
where $\mathcal{P}^{i}$ is the set of pairs $(\mathbb{Q},S)$ consisting of a
probability measure~$\mathbb{Q}$ and an $\mathbb{R}^{d}$-valued martingale~$S$
under~$\mathbb{Q}$ satisfying the conditions of
Theorem~\ref{th:2012-10-03:ftap} and such that $S_{t}^{i}=1$ for each
$t=0,\ldots,T$.

In Theorem~6.1 of \cite{Loehne_Rudloff2014} this result was proved under the
so-called robust no-arbitrage condition of \cite{schachermayer2004} and
subject to the simplifying assumption that the solvency cone $\mathcal{K}_{t}$
contains no lines for any~$t$ (that is, the transaction costs are non-zero for
any~$t$). Their proof is based on the scalarisation procedure of
\cite{Hamel_Heyde2010} for the dual representation of the set $SHP_{0}$ of
superhedging portfolios.

By comparison, the result in \cite{Roux_Zastawniak2014} is free of these
restrictions: it works under the assumption that there is no arbitrage
opportunity as defined in Section~\ref{Sect:curr-mod}, which is weaker than
the robust no-arbitrage condition, and without the need to assume that the solvency cone
$\mathcal{K}_{t}$ contain no lines. It is also a much more general result that
applies to American type derivatives, which reduces to~(\ref{Eq:price-repr})
for European options. The proof is based on the dual construction from
Section~\ref{Sect:primal-dual-constr}, which can in fact be used to produce a
pair $(\mathbb{Q},S)$ that realises the supremum in~(\ref{Eq:price-repr})
(though in general such a pair does not lie in $\mathcal{P}^{i}$ as
$\mathbb{Q}$ may be a degenerate measure, absolutely continuous with respect
to but not necessarily equivalent to~$\mathbb{P}$).%

\section{Conclusions}

We have established a close link, indeed an equivalence between the three
approaches: the above primal and dual constructions and the SHP-algorithm of
\cite{Loehne_Rudloff2014}. The primal construction involves primal objects
only. The dual construction deals exclusively with dual objects (support
functions). Meanwhile, the SHP-algorithm {switches} back and forth
between primal and dual objects (in this case the lower images of the dual
problem~(\ref{eq:D*})). By Proposition~\ref{Prop1-ver2}, these two types of
dual objects are in one-to-one correspondence, which means that the apparent
differences between the algorithms are merely superficial.

Moreover, all three approaches lend themselves well to numerical
implementation: the primal and dual constructions utilise available software
libraries for handling convex sets, whereas the SHP-algorithm makes an
innovative use of Benson's procedure.
In both approaches the procedure limiting computational efficiency is vertex enumeration. An advantage offered by Benson's algorithm is the ability to control the accuracy versus efficiency by choosing an error level. On the other hand, the \emph{Maple} package \emph{Convex} used by \cite{Roux_Zastawniak2014} employs exact arithmetic with rational numbers, hence there is no rounding beyond the conversion (as accurate as one needs it to be) of input data from real to rational numbers. While accurate rational arithmetic carries obvious computational overheads, the primal and dual algorithms are efficient enough so this does not become a problem in realistic multi-step and multi-asset examples that have been investigated, where the computation times were of the order of a couple of minutes on a standard PC machine.

One major difference as compared with the SHP-algorithm approach is that the primal and dual constructions have been
developed in \cite{Roux_Zastawniak2014} for {the} much wider class of American
type options, {and can} handle early exercise problems. In this
context, European options are a particularly straightforward special case. It
remains an open question whether or not the SHP-algorithm of
\cite{Loehne_Rudloff2014} could be extended to American options, at least in
the case of hedging and pricing the seller's position. It would be exciting to
see this happen.

On the other hand, there are limits to what can be expected of the
SHP-algorithm. American options present a particular obstacle that this
approach is unlikely to be able to overcome. Namely, the case of hedging and
pricing the buyer's (rather than the seller's) position in an American option
leads to a non-convex optimisation problem, which is unlikely to yield to the
power of linear vector optimisation methods and geometric duality. For the
same reason, the dual construction collapses as there are no convex dual
objects to work with in the first place. {Nonetheless}, the primal construction can
still be adapted to handle this case; see Example~7.1 in \cite{Roux_Zastawniak2014} for details.

\end{document}